\documentclass[conference]{IEEEtran}
\usepackage[cmex10]{amsmath}
\usepackage{amssymb}
\newtheorem{theorem}{Theorem}

\newtheorem{proposition}{Proposition}

\newcommand{\displayproblem}[3]
{%
\begin{trivlist}%
\item
\begin{list}{}{\setlength{\leftmargin}{5.5em}\setlength{\rightmargin}{1em}\setlength{\labelwidth}{\leftmargin}}%
\item[\bf Problem:] {\sc #1}
\item[\bf Instance:] #2
\item[\bf Question:] #3
\end{list}%
\end{trivlist}%
}
\begin{document}
\title{On the Hardness of Entropy Minimization and Related Problems}
\author{\IEEEauthorblockN{Mladen Kova\v cevi\' c, Ivan Stanojevi\' c, and Vojin \v Senk}
        \IEEEauthorblockA{Department of Electrical Engineering, Faculty of Technical Sciences, \\
                          University of Novi Sad, 21000 Novi Sad, Serbia \\
                          Emails: \{kmladen, cet\_ivan, vojin\_senk\}@uns.ac.rs }}
%\thanks{Date: July 05, 2012.}%
\maketitle
\begin{abstract}
  We investigate certain optimization problems for Shannon information 
  measures, namely, minimization of joint and conditional entropies 
  $H(X,Y)$, $H(X|Y)$, $H(Y|X)$, and maxi\-mization of mutual information 
  $I(X;Y)$, over convex regions. When restricted to the so-called transportation 
  polytopes (sets of distributions with fixed marginals), very simple proofs 
  of NP-hardness are obtained for these problems because in that case they 
  are all equivalent, and their connection to the well-known \textsc{Subset sum} 
  and \textsc{Partition} problems is revealed. The computational intractability 
  of the more general problems over arbitrary polytopes is then a simple 
  consequence. Further, a simple class of polytopes is shown over which 
  the above problems are not equivalent and their complexity differs sharply, 
  namely, minimization of $H(X,Y)$ and $H(Y|X)$ is trivial, while minimization 
  of $H(X|Y)$ and maximization of $I(X;Y)$ are strongly NP-hard problems. 
  Finally, two new (pseudo)metrics on the space of discrete probability 
  distributions are introduced, based on the so-called variation of 
  information quantity, and NP-hardness of their computation is shown. 
\end{abstract}
%
%\IEEEpeerreviewmaketitle
%
\section{Introduction}
Joint entropy $H(X,Y)$, conditional entropies $H(X|Y)$, $H(Y|X)$, and mutual 
information $I(X;Y)$, 
are some of the founding concepts of information theory. In the present paper 
we investigate some natural optimization problems associated with these 
functionals, namely, minimization of joint and conditional entropies and 
maximization of mutual information over convex polytopes, and show that all 
these problems are NP-hard. Certain special cases of these problems are found 
to represent information theoretic analogues of the well-known \textsc{Subset sum} 
and \textsc{Partition} problems. Our results will thus provide a simple, yet 
interesting connection between complexity theory and information theory. 
\par Various optimization problems for the above mentioned information measures 
are studied in the literature. An important example is the well-known 
\emph{Maximum entropy principle} formulated by Jaynes \cite{jaynes}, which states 
that, among all proba\-bility distributions satisfying certain constraints 
(expressing our knowledge about the system), one should pick the one with maximum 
entropy. It has been recognized by Jaynes, as well as many other researchers, that 
this choice gives the least biased, the most objective distribution consistent 
with the information one possesses about a system. Maximizing entropy under 
constraints is therefore an important problem, and it has been thoroughly 
studied (see, e.g., \cite{harr,kapur1}). 
\par It has also been argued \cite{kapur2,yuan} that minimum entropy distributions 
can be of as much interest as maximum entropy distributions. 
The MinMax information measure, for example, has been introduced in \cite{kapur2} 
as a measure of the amount of information contained in a given set of constraints, 
and it is based both on maximum and 
minimum entropy distributions. More generally, entropy minimization is also very 
important conceptually. Watanabe \cite{pattern} has shown that many algorithms 
for clustering and pattern recognition can be characterized as suitably defined 
entropy minimization problems. 
\par Since entropy is a concave\footnote{\,To avoid possible confusion, concave 
means $\cap$ and convex means $\cup$.} functional, its maximization can be solved 
by standard concave maximization methods. On the other hand, concave minimization 
is in general a much harder problem \cite{onn}. Indeed, we will show that the 
minimization of joint entropy over convex polytopes is NP-hard. In fact, we will 
show that a much more restrictive problem is NP-hard, that of entropy minimization 
over the so-called transportation polytopes, i.e., entropy minimization under constraints 
on the marginal distributions. 
Restricting the problem to transportation polytopes is perhaps the key step in our 
analysis and has several advantages. First, it enables one to obtain a very simple 
proof of NP-hardness by using a reduction from the \textsc{Subset sum} problem and 
some simple information theoretic identities and inequalities. Second, it will 
immediately follow from this proof that the problems of minimization of conditional 
entropies and maximization of mutual information are also NP-hard. This claim 
looks difficult to prove by some other methods because these functionals are neither 
concave nor convex. 
\par Maximization of mutual information is certainly an important problem, studied 
in many different scenarios. A familiar example is computing the capacity of the 
channel which amounts to the maximization of this functional over all input 
distributions. This is again a convex maximization problem for which efficient 
algorithms exist \cite{cover}. In Section \ref{onemarginal} we will show that the 
reverse problem -- maximizing mutual information over conditional distributions, 
given the input distribution -- is NP-hard. Another important example is the so-called 
\emph{Maximum mutual information (MMI) criterion} used in the design of classifiers. 
See, e.g., \cite{battiti,deng} for two important applications of this principle. 
\section{Basic definitions}
This section reviews the definitions and basic properties of the quantities that 
will be used later. 
\subsection{Shannon information measures}
Shannon entropy of a random variable $X$ with probability distribution $P=(p_i)$ 
is defined as: 
\begin{equation}
  H(X) \equiv H(P) = -\sum_{i} p_i\log p_i 
\end{equation}
with the usual convention $0\log0=0$ being understood. For a pair of random 
variables $(X,Y)$ with joint distribution $S=(s_{i,j})$ and marginal 
distributions $P=(p_i)$ and $Q=(q_j)$, the following defines their joint entropy: 
\begin{equation}
  H(X,Y) \equiv H_{X,Y}(S) = -\sum_{i,j} s_{i,j}\log s_{i,j}, 
\end{equation}
conditional entropy: 
\begin{equation}
  H(X|Y) \equiv H_{X|Y}(S) = -\sum_{i,j} s_{i,j}\log \frac{s_{i,j}}{q_j}, 
\end{equation}
and mutual information: 
\begin{equation}
  I(X;Y) \equiv I_{X;Y}(S) = \sum_{i,j} s_{i,j}\log \frac{s_{i,j}}{p_{i}q_{j}}, 
\end{equation}
again with appropriate conventions. All of these quantities are related by 
simple identities: 
\begin{equation}
\label{identity}
 \begin{aligned}
   H(X,Y) &= H(X) + H(Y) - I(X;Y)  \\ 
          &= H(X) + H(Y|X) 
 \end{aligned}
\end{equation}
and obey the following inequalities: 
\begin{align}
\label{ineqH}
   \max\big\{H(X), H(Y)\big\}\leq H(X&,Y) \leq H(X) + H(Y), \\
\label{ineqI}
  \min\big\{H(X), H(Y)\big\} &\geq I(X;Y) \geq 0, \\
\label{ineqHx}
  0 \leq H(X|Y) &\leq H(X). 
\end{align}
Equalities on the right-hand sides of \eqref{ineqH}--\eqref{ineqHx} are 
achieved if and only if $X$ and $Y$ are independent. Equalities on the 
left-hand sides of \eqref{ineqH} and \eqref{ineqI} are achieved if and 
only if $X$ deterministically depends on $Y$, or vice versa. Another way 
to put this is that their joint distribution (written as a matrix) has at 
most one nonzero entry in every row, or in every column. Equality on the 
left-hand side of \eqref{ineqHx} holds if and only if $X$ deterministically 
depends on $Y$. We will use these properties in our proofs. For their 
demonstration we point the reader to the standard reference \cite{cover}. 
\par From identities \eqref{identity} one can make the following simple, 
but crucial, observation: Over a set of two-dimensional probability 
distributions with fixed marginals (and hence fixed marginal entropies 
$H(X)$ and $H(Y)$), all the above functionals differ up to an additive 
constant (and a minus sign in the case of mutual information). This means 
in particular that the minimization of joint entropy over such domains is 
equivalent to the minimization of either one of the conditional entropies, 
or to the maximization of mutual information. Therefore, NP-hardness of 
any of these problems will imply that all of them are NP-hard. And finally, 
this will imply that more gene\-ral problems of minimization/maximization 
of the corresponding functionals over arbitrary convex polytopes are NP-hard. 
\subsection{Transportation polytopes}
\label{transport}
Let $\Gamma_{n}^{(1)}$ and $\Gamma_{n\times m}^{(2)}$ denote the sets of one- and 
two-dimensional probability distributions with alphabets of size $n$ and $n\times m$, 
respectively: 
\begin{align}
  \Gamma_{n}^{(1)} &= \big\{(p_i)\in\mathbb{R}^{n}\,:\,p_i\geq0\,,\,\sum_i p_i = 1\big\}  \\ 
  \Gamma_{n\times m}^{(2)} &= \big\{(p_{i,j})\in\mathbb{R}^{n\times m}\,:\,p_{i,j}\geq 0\,,\,\sum_{i,j} p_{i,j} = 1\big\}
\end{align}
Now consider the set of all distributions with marginals $P\in\Gamma_{n}^{(1)}$ 
and $Q\in\Gamma_{m}^{(1)}$, denoted $\mathcal{C}(P,Q)$: 
\begin{equation}
  \mathcal{C}(P,Q) = \left\{S\in\Gamma_{n\times m}^{(2)}\,:\,\sum_j s_{i,j}=p_i,\, \sum_i s_{i,j}=q_j\right\} 
\end{equation}
(letter $\mathcal{C}$ stands for coupling). It is easy to show that sets 
$\mathcal{C}(P,Q)$ are convex and closed in $\Gamma_{n\times m}^{(2)}$. They are 
also clearly disjoint and cover entire $\Gamma_{n\times m}^{(2)}$, i.e., they 
form a partition of $\Gamma_{n\times m}^{(2)}$. Finally, they are parallel affine 
$(n-1)(m-1)$-dimensional subspaces of the $(n\cdot m-1)$-dimensional space 
$\Gamma_{n\times m}^{(2)}$. (We of course have in mind the restriction of the 
corresponding affine spaces in $\mathbb{R}^{n\times m}$ to $\mathbb{R}_{+}^{n\times m}$. ) 
\par The set of distributions with fixed marginals is basically the set of 
nonnegative matrices with prescribed row and column sums (only now we require 
the total sum to be one, but this is inessential). Such sets are known in 
discrete mathematics as transportation polytopes \cite{brualdi}. Their 
name comes from the fact that they correspond to the following problem: given 
$n$ supplies $p_1,\ldots,p_n$ and $m$ demands $q_1,\ldots,q_m$ of some "goods" 
(total supply and total demand being equal), describe all ways of transporting 
the goods so that the demands are fulfilled. For example, one possible solution 
to the transportation problem with $P=(1,3,5)$ and $Q=(2,4,3)$ would be 
\begin{equation*}
  \begin{pmatrix}
       1 & 0 & 0  \\ 
       1 & 2 & 0  \\ 
       0 & 2 & 3 
 \end{pmatrix}
\end{equation*}
which is obtained by the so-called north-west corner rule. The set of all possible 
solutions constitutes a polytope ${\mathcal C}(P,Q)$. 
\par In the context of the problems studied here, the following facts will be useful. 
A concave continuous function over any bounded convex polytope must attain its 
minimum over one of the vertices of the polytope. An interesting fact about 
transportation polytopes is that, if the marginals are integer, that is if 
$\,p_i,q_j\in\mathbb{N}$, then all the vertices (observed as matrices) have only 
integer entries. As a consequence, if the marginals are rational, 
that is if $\,p_i,q_j\in\mathbb{Q}$, then all the vertices have only rational entries. 
Furthermore, it is not hard to see that the description length\footnote{\,By the 
description length of an object, we mean the number of bits required to write it 
down, as usual in the context of algorithmic problems.} of the vertices is 
polynomial in the description length of the marginals. We shall make use of these 
facts in the proofs of Theorems \ref{minentropy} and \ref{channel}. 
\section{Entropy over transportation polytopes}
\label{entropytp}
As we noted before, we will focus here on sets of probability distributions with 
fixed marginals, i.e., we will consider the above mentioned optimization problems 
over transportation polytopes. The problems turn out to be NP-hard even under this 
restriction, and this is perhaps the easiest way to prove NP-hardness of their more 
general versions. 
\par Let some marginal distributions $P$ and $Q$ be given, and observe 
${\mathcal C} (P,Q)$. From identities \eqref{identity} one sees that over 
${\mathcal C} (P,Q)$ the minimization of $H(X,Y)$ is equivalent to the minimization 
of $H(X|Y)$ and $H(Y|X)$, or to the maximization of $I(X;Y)$, so it is enough 
to consider only the joint entropy for example. Joint entropy $H(X,Y)$ is well-known 
to be concave in the joint distribution, and so its minimization belongs to a 
wide class of concave minimization problems which are in general intractable 
\cite{onn}. Conditional entropies $H(X|Y)$ and $H(Y|X)$ are neither concave nor 
convex in the joint distribution, but over ${\mathcal C} (P,Q)$ they \emph{are} 
concave because in that case they differ from the joint entropy only by an 
additive constant. By the same reasoning, mutual information is convex in the 
joint distribution over ${\mathcal C} (P,Q)$. Based on concavity one concludes 
that the optimizing distribution for these problems must be one of the vertices 
of ${\mathcal C} (P,Q)$. The trouble with concave functions, of course, is that 
one must visit all of them, or at least a "large" portion of them, to decide 
where the minimum is. 
\subsection{The computational problems}
The most general form of the problem in the context studied here would be the 
following: Given a polytope (by a system of inequalities, say) in 
$\Gamma_{n\times m}^{(2)}$, find the distribution (matrix) $S$ which minimizes 
the entropy functional $H_{X,Y}(S)$. The decision version of this problem is 
obtained by giving some threshold $h$ at the input, and asking whether a given 
polytope contains a distribution $S$ with $H_{X,Y}(S)\leq h$. 
\par Let us now restrict the problem to transportation polytopes. Since: 
\begin{equation}
  H(X,Y) \geq \max\big\{H(X),H(Y)\big\}, 
\end{equation}
over the transportation polytope ${\mathcal C} (P,Q)$ we have: 
\begin{equation}
\label{lower}
  H(X,Y) \geq \max\big\{H(P),H(Q)\big\} 
\end{equation}
with equality if and only if $Y$ determi\-nistically depends on $X$, or vice 
versa \cite{cover}. In other words, we will have equality in \eqref{lower} 
iff the joint distribution is such that it has at most one nonzero entry in 
every row, or in every column. We can now formulate an even more restrictive 
problem, with the threshold specified in advance: Given a transportation 
polytope ${\mathcal C} (P,Q)\in\Gamma_{n\times m}^{(2)}$, is there a distribution 
$S$ in this polytope with $H_{X,Y}(S)\leq H(P)$? Note that, because of 
\eqref{lower}, this inequality must in fact be an equality. We name this problem 
\textsc{Entropy minimization} even though the name would probably be more 
appropriate for the more general problems mentioned above. 
%
% \begin{description}
%  \item[Problem] \textsc{Entropy minimization}
%  \item[Instance] Positive rational numbers $p_1,\ldots,p_n$ and 
%                  $q_1,\ldots,q_m$, with $\sum_{i=1}^{n} p_i = \sum_{j=1}^{m} q_j = 1$.
% \end{description}
% 
\displayproblem
  {Entropy minimization}
  {Positive rational numbers $p_1,\ldots,p_n$ and $q_1,\ldots,q_m$, 
   with $\sum_{i=1}^{n} p_i = \sum_{j=1}^{m} q_j = 1$.}
  {Is there a matrix $S\in\mathcal{C}(P,Q)$ with entropy
   $H_{X,Y}(S) = H(P)$?}
%    \begin{tabbing}
%     \indent
%     \textbf{Problem}: \textsc{Entropy minimization}  \\ 
%     \indent
%     \textbf{Instance}: \= Positive rational numbers $p_1,\ldots,p_n$ and  \\ 
%                        \> $q_1,\ldots,q_m$, with $\sum_{i=1}^{n} p_i = \sum_{j=1}^{m} q_j = 1$. \\
%     \indent
%     \textbf{Question}: \= Is there a matrix $S\in\mathcal{C}(P,Q)$ with entropy \\ 
%                        \> $H_{X,Y}(S) = H(P)$ ? 
%    \end{tabbing}
%
\noindent This problem will be shown to be NP-complete. 
\par We also briefly note here that the corresponding problem of the maximization 
of $H(X,Y)$ (maximization of $H(X|Y)$ or minimization of $I(X;Y)$) over 
${\mathcal C} (P,Q)$ is trivial because: 
\begin{equation}
  H(X,Y) \leq H(X) + H(Y) 
\end{equation}
with equality if and only if $X$ and $Y$ are independent \cite{cover}, i.e., 
iff their joint distribution is $P\times Q=(p_iq_j)$, and this distribution 
clearly belongs to ${\mathcal C} (P,Q)$. 
\subsection{The proof of NP-hardness}
We describe first the \textsc{Subset sum}, a well-known NP-complete problem 
\cite{garey}, which will be the basis of the proof to follow: 
 \displayproblem
   {Subset sum}
   {Positive integers $d_1,\ldots,d_n$ and $s$.}
   {Is there a $J\subseteq\{1,\ldots,n\}$ such that $\sum_{j\in J} d_j = s$ ?}
\begin{theorem}
\label{minentropy}
  \textsc{Entropy minimization} is NP-complete. 
\end{theorem}
\begin{proof}
  We shall demonstrate a reduction from the \textsc{Subset sum} problem to the 
  \textsc{Entropy minimization} problem. Let there be given an instance of the 
  \textsc{Subset sum} problem, i.e., a set of positive integers 
  $s\,,\,d_1,\ldots,d_n$, $n\geq 2$. 
  Let $D=\sum_{i=1}^{n} d_i$, and let $p_i=d_i/D$, $q=s/D$. The question we are 
  trying to answer is whether there is a $J\subseteq\{1,\ldots,n\}$ such that 
  $\sum_{j\in J} d_j = s$. Observe that this is equivalent to asking whether 
  there is a matrix $S$ with row sums $P=(p_1,\ldots,p_n)$ and column sums 
  $Q=(q,1-q)$, which has at most one nonzero entry in every row (or, in 
  probabilistic language, such that $Y$ deterministically depends on $X$). 
  We know that in this case, and only in this case, the entropy of $S$ would 
  be equal to $H(P)$ \cite{cover}. 
  So if we create an instance of the \textsc{Entropy minimization} problem with 
  $P$ and $Q$ as above, the answer to the question whether there exists 
  $S\in\mathcal{C}(P,Q)$ with $H_{X,Y}(S) = H(P)$ will solve the 
  \textsc{Subset sum} problem. Therefore, this is the reduction we wanted. 
  It is left to prove that \textsc{Entropy minimization} belongs to NP. 
  This is done by using the familiar characterization of the class NP 
  via certificates \cite{pap}. We have to show that every \textsc{Yes}-instance 
  of the problem has a succinct certificate, while no \textsc{No}-instance has  
  one, and that the validity of the alleged certificates can be verified in 
  polynomial time. The certificate is of course the optimizing distribution itself. 
  That it is succinct is easy to show (see the comment in the last paragraph of 
  Section \ref{transport}), and polynomial time verifiability is even easier, 
  because we only have to check that $S$ belongs to $\mathcal{C}(P,Q)$ and that 
  it has at most one nonzero entry in every row. 
\end{proof}
\par As a straightforward consequence of the above claim, the more general problem 
of finding a minimizing distribution is NP-hard. It is an interesting task to 
determine the precise complexity of this problem (in the sense of proving that 
it is complete for some natural complexity class). Note that even determining 
whether it belongs to FNP\footnote{\,The class FNP captures the complexity of 
function problems associated with decision problems in NP, see \cite{pap}.} 
is nontrivial. Whether the 
decision version of this problem, namely, deciding whether a given polytope 
contains a distribution with entropy smaller than a given threshold, belongs to 
NP is also an interesting question (which we shall not be able to resolve here). 
One has to be careful when reasoning about "certificates" for these problems. 
Namely, one has to able to check in polynomial time that the certificate is 
indeed valid. In the above proof, we only had to check that the given matrix 
(the alleged certificate) belongs to $\mathcal{C}(P,Q)$ (i.e., that it has 
nonnegative entries and prescribed row and column sums) and that it has at most 
one nonzero entry in every row, and this is clearly easy to do. 
But in the more general problems mentioned above, one is required to compute 
numbers of the form $a\log a$ to check whether $H(T) \leq h$ for example. These 
numbers are in general irrational, and therefore verifying this inequality might 
not be computationally trivial as it might seem. It is interesting to mention in 
this context the so-called \textsc{Sqrt sum} problem: 
 \displayproblem
   {Sqrt sum}
   {Positive integers $d_1,\ldots,d_n$, and $k$.}
   {Decide whether $\sum_{i=1}^n \sqrt{d_i} \leq k$ ?}
This problem, though "conceptually simple" and bearing certain resemblance with 
checking of certificates in the general versions of the entropy minimization 
problem, is not known to be solvable in NP \cite{etessami} (it is solvable in PSPACE). 
\section{Two pseudometrics which are hard to compute}
A variant of the minimization of the above mentioned quantities produces a distance 
on the space of discrete proba\-bility distributions. For a pair of random variables 
$(X,Y)$ with joint distribution $S$, define \cite{csiszar}:
\begin{equation} 
 \begin{split}
   \Delta(X,Y) &\equiv \Delta(S)\! \begin{aligned}[t] &= H(X,Y) - I(X;Y)  \\
                                                    &= H(X|Y) + H(Y|X)
                                 \end{aligned}  \\ 
   \Delta'(X,Y) &\equiv \Delta'(S) = 1 - \frac{I(X;Y)}{H(X,Y)}
 \end{split} 
\end{equation}
The quantity $\Delta(X,Y)$ is sometimes called the \emph{variation of information}. 
Its normalized variant, $\Delta'(X,Y)$, is basically an information theoretic analogue 
of the Jaccard distance between finite sets. 
Both of these quantities satisfy the properties of a 
pseudometric \cite{csiszar}. However, when this statement is made, one must assume 
that the joint distribution of $(X,Y)$ is given because joint entropy and mutual 
information are not defined otherwise. This is usually overlooked in the literature. 
Furthermore, if these quantities are used as distance measures on the space of all 
random variables, then joint distributions of every pair of random variables must 
be given. For example, one could first define some random process $(X_t)$ and then 
take $\Delta$ or $\Delta'$ as distances between the random variables $X_t$.  
In order to avoid the dependence on the chosen random process (or on some universal 
joint distribution), and to define a distance between individual random variables 
(more precisely, between their distributions) one can make the following definitions: 
\begin{equation}
 \begin{aligned}
       \underline{\Delta}(P,Q) &= \inf_{S\in{\mathcal{C}}(P,Q)} \big\{H_{X,Y}(S) - I_{X;Y}(S)\big\}, \\ 
       \underline{\Delta}'(P,Q) &= \inf_{S\in{\mathcal{C}}(P,Q)} \left\{1 - \frac{I_{X;Y}(S)}{H_{X,Y}(S)}\right\}. 
 \end{aligned}
\end{equation}
This definition mimics the one for the total variation distance: 
\begin{equation}
  d_\text{V}(X,Y) = \inf_{{\mathcal{C}}(P,Q)} \big\{{\mathbb{P}}(X\neq Y)\big\} 
\end{equation}
where the infimum is taken over all joint distributions of the random vector $(X,Y)$ 
with marginals $P$ and $Q$. 
\par Let $\Gamma^{(1)}=\big\{(p_i)_{i\in\mathbb{N}}\,:\,p_i\geq0\,,\,\sum_i p_i = 1\big\}$. 
We have the following. 
\begin{proposition}
  $\underline{\Delta}$ and $\underline{\Delta}'$ are pseudometrics on $\Gamma^{(1)}$. 
\end{proposition}
The proof of this proposition is not difficult but we omit it here since it is 
not essential for our current aims. We can now prove one more intractability 
result. 
\begin{theorem}
  Given rational $P$ and $Q$, determining whether $\underline{\Delta}(P,Q) = H(P) - H(Q)$ is NP-hard. 
\end{theorem}
\begin{proof}
  Note that 
  \begin{equation}
   \begin{aligned}
    \underline{\Delta}(P,Q) &= \inf_{S\in{\mathcal{C}}(P,Q)} \big\{H_{X,Y}(S) - I_{X;Y}(S)\big\}  \\ 
                       &= 2\inf_{S\in{\mathcal{C}}(P,Q)} \big\{H_{X,Y}(S)\big\} - H(P) - H(Q).
   \end{aligned}
  \end{equation}
  Now the claim follows directly from Theorem \ref{minentropy}. 
\end{proof}
\section{One marginal fixed}
\label{onemarginal}
In this section we address similar problems as before, only now we fix only one 
of the marginal distributions, say $P=(p_1,\ldots,p_n)$. If the cardinality of 
the alphabet of the other random variable $Y$ is not specified, then the problems 
are trivial. Namely, one takes $Q=P$ and for $S=\text{diag}(P)$ (two-dimensional 
distribution with masses $p_i$ on the diagonal and zeros elsewhere) one has 
$H_{X,Y}(S)=I_{X;Y}(S)=H(P)$, and hence $S$ is optimal. So assume that the 
cardinality of the other alphabet is bounded to $m$. Denote the set of all 
distributions with marginal distribution of $X$ fixed to $P$ and the cardinality 
of the alphabet of $Y$ fixed to $m$, by ${\mathcal C}(P,m)$. We have 
\begin{equation}
  {\mathcal C}(P,m) = \bigcup_{Q\in\Gamma_{m}^{(1)}} {\mathcal C}(P,Q). 
\end{equation}
\par Minimization of the joint entropy $H(X,Y)$ over such polytopes is trivial. The 
reason is that $H(X,Y)\geq H(P)$ with equality iff $Y$ deterministically depends 
on $X$, and so the solution is \emph{any} joint distribution having at most one 
nonzero entry in each row. Since $H(X)$ is fixed, this also minimizes the 
conditional entropy $H(Y|X)$. 
The other two optimization problems considered so far, minimization of 
$H(X|Y)$ and maximization of $I(X;Y)$, are still equivalent because 
$I(X;Y)=H(X)-H(X|Y)$, but they turn out to be much harder. Therefore, in the 
following we shall consider only the maximization of $I(X;Y)$. 
\par When one marginal is fixed, choosing the optimal joint distribution amounts 
to choosing the optimal conditional distribution $p(y|x)$. Mutual information 
$I(X;Y)$ is known to be convex in the conditional distribution \cite{cover} 
(and hence, $H(X|Y)$ is concave in $p(y|x)$, for fixed $p(x)$) and so this 
is again a convex maximization problem. This conditional distribution can be 
thought of as a discrete memoryless communication channel with $n$ input 
symbols and $m$ output symbols, and hence we name the corresponding 
computational problem \textsc{Optimal channel}. 
 \displayproblem
   {Optimal channel}
   {Positive rational numbers $p_1,\ldots,p_n$ with $\sum_{i=1}^{n} p_i = 1$, 
    and an integer $m$.}
   {Is there a channel $C\in{\mathcal C}(P,m)$ with mutual information $I_{X;Y}(C)\geq\log m$ ?}
Note that the above inequality must in fact be an equality because over 
${\mathcal C}(P,m)$: 
\begin{equation}
\label{upper}
  I(X;Y) \leq \min\big\{H(P),\log m\big\} 
\end{equation}
which follows from \eqref{ineqI} and the fact that $H(Y)\leq\log m$. The above 
problem is a suitable restriction of a more general problem of finding a maximizing 
distribution, as we did with {\sc Entropy minimization}. 
\subsection{The proof of NP-hardness}
We describe next the well-known \textsc{Partition} (or \textsc{Number partitioning}) 
problem \cite{garey}. 
 \displayproblem
   {Partition}
   {Positive integers $d_1,\ldots,d_n$.}
   {Is there a partition of $\{d_1,\ldots,d_n\}$ into two subsets with equal sums?}
This is clearly a special case of the \textsc{Subset sum} problem. It can be solved in 
pseudo-polynomial time by dynamic programming methods \cite{garey}. But the following 
closely related problem is much harder. 
 \displayproblem
   {3-Partition}
   {Nonnegative integers $d_1,\ldots,d_{3m}$ and $k$ with $k/4<d_j<k/2$ and $\sum_{j} d_j = mk$.}
   {Is there a partition of $\{1,\ldots,3m\}$ into $m$ subsets $J_1,\ldots,J_m$ 
    (disjoint and covering $\{1,\ldots,3m\}$) such that $\sum_{j\in J_r} d_j$ are 
     all equal? (The sums are necessarily $k$ and every $J_i$ has $3$ elements.)}
This problem is NP-complete in the strong sense \cite{garey}, i.e., no 
pseudo-polynomial time algorithm for it exists unless P=NP. 
\par The following theorem will establish that, given an information source, 
determining the best channel (in the sense of having the largest mutual 
information) is NP-hard. 
\begin{theorem}
 \label{channel}
  \textsc{Optimal channel} is NP-complete. 
\end{theorem}
\begin{proof}
  We prove the claim by reducing 3-\textsc{Partition} to \textsc{Optimal channel}. 
  Let there be given an instance of the 3-\textsc{Partition} problem as described 
  above, and let $p_i=d_i/D$ where $D=\sum d_i$. Deciding whether there exists a 
  partition with described properties is clearly equivalent to deciding whether 
  there is a matrix $C\in\mathcal{C}(P,\,m)$ with the other marginal $Q$ being 
  uniform and $C$ having at most one nonzero entry in every row (i.e., $Y$ 
  deterministically depending on $X$). This on the other hand happens if and only if 
  there is a matrix $C\in\mathcal{C}(P,\,m)$ with mutual information 
  equal to $H(Q)=\log m$. Therefore, solving the \textsc{Optimal channel} problem 
  with instance $(p_i)$ as above will solve the 3-\textsc{Partition} problem. 
  This shows the NP-hardness of \textsc{Optimal channel}. It is left to prove that 
  it belongs to NP. The reasoning here is completely analogous to the one in the 
  proof of Theorem \ref{minentropy}, namely, the certificate is the optimal 
  distribution/matrix itself. 
\end{proof}
\par The problem remains NP-complete even over ${\mathcal C}(P,\,2)$, i.e., 
when the cardinality of the channel output is fixed in advance to $2$. In that 
case the problem is equivalent to the \textsc{Partition} problem. 
\par It is easy to see that the transformation in the proof of Theorem 
\ref{channel} is in fact \emph{pseudo-polynomial} \cite{garey} which implies 
that \textsc{Optimal channel} is strongly NP-complete and, unless P=NP, has no 
pseudo-polynomial time algorithm. 
\section*{Acknowledgment}
  The authors would like to acknowledge the financial support of the Ministry 
  of Science and Technological Development of the Republic of Serbia 
  (grants No. TR32040 and III44003). 

\end{document}